\RequirePackage[2020-04-06]{latexrelease}

\documentclass[letterpaper, 10 pt, conference]{cssconf}  

\IEEEoverridecommandlockouts                              
\pdfoutput=1

\let\labelindent\relax

\usepackage{enumitem}
\usepackage{times}
\usepackage{amsmath, amssymb, amsthm} 
\usepackage{graphicx}

\include{cmds}

\title{\LARGE \bf%
	On the complexity of linear systems: an approach via \\ rate distortion theory and emulating systems
	}

\author{Eric Wendel, John Baillieul, and Joseph Hollmann
	\thanks{E. Wendel is with the Division of Systems Engineering, Boston University, Boston MA 02215 USA, and also with The Charles Stark Draper Laboratory, Inc., Cambridge, MA 02139 USA {\tt\small edbw@bu.edu}.}%
	\thanks{J. Baillieul is with the Division of Systems Engineering, Boston University, Boston MA 02215 USA {\tt\small johnb@bu.edu}.}%
	\thanks{J. Hollmann is with The Charles Stark Draper Laboratory, Inc., Cambridge, MA 02139 USA {\tt\small jhollmann@draper.com}.}%
	}

\makeatletter
\let\NAT@parse\undefined
\makeatother

\begin{document}

	\maketitle
	
	\begin{abstract}
		We define the \textit{complexity} of a continuous-time linear system to be the minimum number of bits required to describe its forward increments to a desired level of fidelity, and compute this quantity using the rate distortion function of a Gaussian source of uncertainty in those increments. The complexity of a linear system has relevance in control-communications contexts requiring local and dynamic decision-making based on sampled data representations. We relate this notion of complexity to the design of attention-varying controllers, and demonstrate a novel methodology for constructing source codes via the endpoint maps of so-called \textit{emulating systems}, with potential for non-parametric, data-based simulation and analysis of unknown dynamical systems.
	\end{abstract}

	\section{Introduction}


	In certain application contexts, for example distributed control, a stabilizing control signal is transmitted from a control subsystem to an open loop plant over a physical communications channel of finite capacity. This capacity can be shared with other sensors and subsystems, and thus both the number of bits and the time allocated to the control subsystem for processing those bits are limited. The celebrated data rate theorem \cite{baillieulFeedbackDesignsControlling1999, wingshingwongSystemsFiniteCommunication1999} established the minimum required channel capacity, in bits per second, below which a controller cannot stabilize an unstable plant. This result was extended through the use of topological feedback entropy to nonlinear systems in \cite{nairTopologicalFeedbackEntropy2004}. See \cite{nairFeedbackControlData2007} for a review of results on minimum required channel capacities for stability. 

	The data rate theorem and related results are focused on the transmission of control signal information for the singular purpose of system stabilization. However, even within a single system information can be transmitted across physical channels for a plurality of tasks. For example, a rate sensor with a fixed bit depth and sampling frequency necessarily provides information about changes in its sensed quantities at a fixed data rate. These bits may require further compression and accumulation before transmission to a local state estimation subsystem or embedded unsupervised learning algorithm. Are there particular regions of state space in which these sensors and subsystems will be required to operate at higher rates and resolution? Conversely, where and at what times can these data rates be lowered without impacting overall system performance?
	
	
	These are questions about required information rates for estimation and control tasks in \textit{local} and \textit{dynamic} sampled-data contexts. The data rate theorem is concerned with information required for \textit{asymptotic} stability. We may ask the nuanced question: \textit{how many bits are required to describe the changes in state of a continuous-time linear system to a required level of fidelity?} Our main result is a complete answer to this question for continuous-time linear systems subject to process noise, with fidelity measured in the mean square sense. Specifically, we show that the minimum number of bits is given by the rate distortion function of the source of system uncertainty, and construct explicit source codes using the endpoint maps of so-called \textit{emulating systems}. The resulting encodings have advantages for non-parametric, data-based simulation and analysis.

	\parsec{Related literature} A notion of ``information complexity'' was introduced in \cite{tatikondaControlCommunicationConstraints2004} as the data rate required to achieve a control or estimation task in the asymptotic limit for the stochastic, linear time-invariant (LTI), discrete-time system
	\[ x_{k+1} = F\,x_k + w_k \] where $k = 0, \ldots, T-1$ and $w_k$ is iid process noise. In the companion paper \cite{tatikondaStochasticLinearControl2004} the authors  compute the average information content per unit time contained in length $T$ trajectories of this linear system, under the additional restriction that the information in trajectories cannot be encoded at once as a single set (or block) of $T$ observations of system state, but must be encoded iteratively as each state in the trajectory sequence is received. This led to the definition of the so-called \textit{sequential rate distortion function} (SRDF), which does not have a closed-form expression in the asymptotic limit as $T\to\infty$ \cite{stavrouTimeInvariantMultidimensionalGaussian2020}.
	
	The systems of interest to \cite{tatikondaControlCommunicationConstraints2004, tatikondaStochasticLinearControl2004} could be obtained by sampling a stochastic continuous-time system at a constant uniform rate $f_s = 1/\Delta t$. Our interest is in local information content relevant to dynamic execution of control and estimation tasks at potentially non-uniform sampling rates. We are therefore focused on the information content in forward increments of sampled-data representations of continuous-time linear stochastic systems, where $\Delta t$ is allowed to increase or decrease by small amounts. As such, we allow block encoding of the information content within a (small) interval of time and assume negligible encoding delays.
	
	Finally, although the information content of an infinitesimally short trajectory as  $\Delta t \to 0$ is of important theoretical interest \cite{weissmanDirectedInformationCausal2013}, our main results require finite sampling rates.

	\parsec{Contributions and organization} Our full contributions and the organization of the paper are as follows. Our primary, novel contribution is Proposition \ref{prop:lti-code-rate-dx} in Section \ref{sec:lti-code-rate}, which establishes the minimum amount of information required to describe how the state of a continuous-time linear system changes over a finite time interval as given by the rate distortion function of a Gaussian source of uncertainty in its forward increments. As an apparently fundamental property of the class of linear systems we consider, we call it the \textit{complexity} of the linear system (Definition \ref{defn:complexity}).
	
	Our second main result is a proof, for the time-invariant case, of the intuition that increasing the sampling rate of a linear system reduces the number of bits required to describe that system at each sampling time (Proposition \ref{prop:max-code-rate}, Corollary \ref{cor:sampling-rate-capacity}). In contrast to the data rate theorem, this is a statement regarding the information content of a linear system with relevance to coder-controller design that one can make regardless of the stability of the system matrix.

	Both of these results follow from the standard ``reverse water-filling'' interpretation of the rate distortion function for a multivariate Gaussian source of uncertainty. We review rate distortion theory in Section \ref{sec:rate-distortion-theory}.

	Our secondary contributions are presented in Section \ref{sec:emu-sys}, where we construct source codes for the forward increments of an unknown system using the endpoint map of a so-called \textit{emulating system} \cite{sunNeuromimeticLinearSystems2022}. In contrast with  than encoding state by discretizing state space as in \cite{tatikondaControlCommunicationConstraints2004}, using the ``zooming'' quantizers of \cite{brockettQuantizedFeedbackStabilization2000}, or an enumeration of a finite open cover of a compact subset of state space as in \cite{nairTopologicalFeedbackEntropy2004}, our state encoders are constructed by appropriate quantization of directions in the tangent space. We discuss the code rate of two example emulating systems and demonstrate how they enable simulation of new sample paths of the unknown system without assumptions about the model parameters or process noise characteristics.

	\section{A review of rate distortion theory} \label{sec:rate-distortion-theory}

	Let $X\sim p(x)$ be a source of continuous vector-valued random variables taking values on a subset $\mathsf{V}\subseteq \R^n$. A \textit{reproduction} of $X$ is a random variable $\hat{X} := g \circ f(X)$ where the \textit{compressor} $f: \mathsf{V} \to \mathsf{U}$ maps a realization of the random variable $X$ onto its representation in the set $\mathsf{U}$ and the \textit{decompressor} $g:\mathsf{U}\to \mathsf{V}$ maps the representation back. 
	

	The quality or fidelity of a reproduction is defined in terms of averages of a so-called \textit{distortion function} $\rho(x, g\circ f(x))$ \[ \rho: \mathsf{V} \times \mathsf{V} \to \R_{\geq0} \] quantifying the consequences of reproducing the source via the transformation $\hat{X} = g\circ f(x)$. Rate distortion theory places few restrictions on the choice of distortion function. 
	In this paper we will consider only the $L_2$ norm distortion function
	\begin{equation}
		\rho(x, \hat{x}) = \| x - \hat{x} \|_2^2,
	\end{equation}
	and its average, the \textit{mean square error}. 

	The expected value of $\rho$ depends on the joint density $P(\hat{x}, x)$ between the source and its reproduction. With a fixed and given source $p(x)$, this joint is completely determined by a conditional density function $Q(\hat{x} \given{} x)$, 
	which may be viewed as a statistical characterization of the behavior of the as-yet unknown de/compressor functions $f$ and $g$. The mutual information between the source and its reproduction depends on $Q(\cdot \given{} \cdot)$
	\begin{equation*}
		I(\hat{X} ; X) = H(\hat{X}) - H(\hat{X} \given{} X) = I(X ; \hat{X})
	\end{equation*}
	%
	and the \textit{rate distortion function} $R(D)$ is determined by the $Q(\cdot \given{} \cdot)$ that minimizes that mutual information
	\begin{equation*} \begin{aligned}
		R(D) := \min_{Q(\hat{X} \given{} X)} &I(\hat{X} ; X) \\
		 \text{s.t.} \quad &E[\rho(X, \hat{X})] \leq D
	\end{aligned} \end{equation*}


	\parsec{Block codes} Let the density $p(x)$ be a (vector-valued) \textit{memoryless source}, meaning that any discrete sequence of random variables $X^{(i)} \sim p(x)$, $i=1, \ldots, L$ is iid, for any $L>0$. Suppose that $\hat{X}^L := (\hat{X}^{(1)}, \ldots, \hat{X}^{(L)})$ is a reproduction of $X^L$ with each $\hat{X}^{(i)}$ taking values on a finite set of vectors $\mathcal{V} = \{V_i\}_{i=1}^K \subset \R^n$. 
	The elements of $\mathcal{V}$ are called \textit{codevectors} or \textit{symbols} and the set $\mathcal{V}$ itself a \textit{source code} of \textit{size} $K$ and \textit{blocklength} $L$ with \textit{code rate} \[ R := \frac{1}{L}\log_2(K) \] in units of bits per symbol, or simply: $\mathcal{V}$ is a \textit{$(K,L)$-source code}.

	Fix a particular $(K,L)$-source code $\mathcal{V}$ with de/compressor functions $f, g$. We measure the expected performance of this source code by averaging the chosen distortion function $\rho$ over the given source distribution:
	\begin{equation*}
		\bar{D} := \frac{1}{L}\sum_{i=1}^L E\big[\rho(X^{(i)}, g\circ f(X^{(i)}))\big]
	\end{equation*}
	We have exploited the fact that the variables $X^{(i)}$ are iid. If $\bar{D} \leq D$, where $D$ is the maximum allowed distortion, then $\mathcal{V}$ is said to be \textit{admissible}. If $\bar{D} > D$ then it is \textit{inadmissible}.

	The source coding theorem, and its converse, establish the rate distortion function as the minimum possible code rate of any admissible source code. We specialize slightly the statement of the general source coding theorem from \cite{bergerRateDistortionTheory2003} for our purposes.
	\begin{thm}[{\cite[Theorem~7.2.4-5]{bergerRateDistortionTheory2003}}] \label{thm:code-rate}
		Let $X \sim p(x)$ be a memoryless source with maximum admissible distortion $D \geq 0$ and rate distortion function $R(D)$.

		Then, for any $\epsilon > 0$ there exists an admissible source code with average distortion $\bar{D} \leq D+\epsilon$ and rate $R < R(D) + \epsilon$. Conversely, any source code with rate $R < R(D)$ has $\bar{D} > D$ and is inadmissible.
	\end{thm}

	The following well-known result (cf. \cite[Theorem~10.3.3]{coverElementsInformationTheory2005}, \cite[equation~(4.5.21)]{bergerRateDistortionTheory2003}) defines the minimum admissible code rate of a memoryless multivariate Gaussian source. We apply it in Section \ref{sec:lti-code-rate} to compute the code rate of the trajectories of a linear system affected by noise. 

	Let $\log^+(x) := \max\{0, \log(x)\}$, $I$ denote the $n\times n$ identity matrix, and $\lambda_i(\Sigma)$ denote the $i$\textsuperscript{th} eigenvalue of a symmetric positive definite matrix $\Sigma$.

	\begin{lem} \label{lem:gaussian-source}
		The rate distortion function of a memoryless source $X \sim p(x) = N(\mu, \Sigma)$ with maximum mean square distortion $D\geq0$ in units of nats per symbol is given by 
		\begin{equation*}
			R(D) = \frac{1}{2}\sum_{i=1}^n \log^+\Big(\frac{\sigma_i^2}{D_i(\theta)}\Big)
		\end{equation*}
		where $\sigma_i^2 := \lambda_i(\Sigma)$, $D_i(\theta) := \min\{\theta, \sigma_i^2\}$ and $\theta\geq 0$ is chosen so that $D = \sum_{i=1}^n D_i(\theta)$. When $D/n < \min_i \{\sigma_i^2\}$ the rate distortion function can be expressed simply as
		\begin{equation*}
			R(D) = \frac{1}{2}\log\det(\Sigma) - \frac{1}{2}\log\det(\frac{D}{n}I)
		\end{equation*}
	\end{lem} %
	This is the classical ``reverse water-filling'' characterization of the minimum admissible code rate for a memoryless Gaussian source, and captures the intuitive result that no bits need be allocated by an optimal compressor to describe any principal components of the source signal whose variance falls below the ``water-level'' or threshold $\theta$.
	
	To simplify the following exposition we express $R(D)$ in nats per symbol for small admissible distortions satisfying $D/n < \min_i\{\lambda_i(\Sigma)\}$, with the understanding that $R(D)$ for large $D$ is obtained by reverse water-filling.

	\section{Code rate and complexity of linear systems}
	\label{sec:lti-code-rate}

	\begin{figure}[t]
		\centering
		\includegraphics[width=\linewidth]{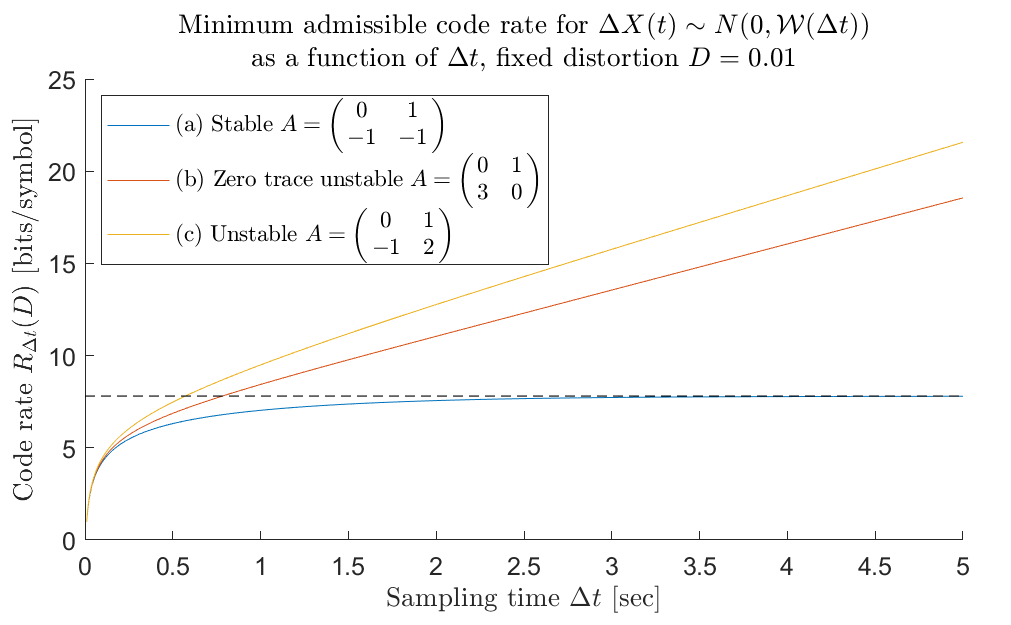}
		\caption{The minimum code rate $R_{\Delta t}(0.01)$ changes with sampling rate $1/\Delta t$ for three different LTI systems with different stability properties. 
		The code rate for all systems starts from 0, corresponding to the trivial change in state $\Delta X(t) = X(t)$. (a) The largest required admissible code rate for this stable LTI system is about $7.8$ bits/symbol, shown by the dashed black horizontal line; (b, c) These systems have no equilibrium solution to (\ref{eqn:lyap-eqn}) and their code rates increase without bound with decreasing sampling rate, corresponding to the accumulation of uncertainty (\ref{eqn:noise-gramian-LTI}) in the change in state due to instability.}
		\label{fig:ex-lti-code-rates-dx}
		\vspace*{-1em}
	\end{figure}

	
	Consider the continuous-time It\^{o} stochastic differential equation
	\begin{equation} \label{eqn:stoch-lin-sys}
		dx(t) = A(t)x(t)\,dt + dw(t),
	\end{equation}
	where $w(t)$ is a $n$-dimensional Brownian motion process with constant covariance $N$. We are interested in the information content in the next sample of system state $X(t+\Delta t)$ given $ X(t)$. Define the ``change of state'' of (\ref{eqn:stoch-lin-sys}) to be the random variable \[ \Delta X(t) := X(t+\Delta t) - X(t) \given{} X(t) \] conditioned on $X(t)$. The change of state is given by variation of constants
	\begin{equation*}
		\Delta X(t) = \Delta \mu(t) + \int^{t+\Delta t}_t\hspace*{-1.5em} \Phi(t+\Delta t, \tau) \,dw(\tau), 
	\end{equation*}
	where the integral on the righthand side is an It\^{o} integral. It follows, cf. \cite[pg.~131]{jazwinskiStochasticProcessesFiltering1970}, that the mean is $\Delta\mu(t) := E[\Delta X(t)] = \big(\Phi(t+\Delta t, t)-I\big)X(t)$ and the covariance $\Cov(\Delta X(t)) = \Cov(X(t+\Delta t))$ is given by the Gramian
	\begin{equation} \label{eqn:noise-gramian}
		\mathcal{W}_t(\Delta t) := \int^{t+\Delta t}_t\hspace*{-1.5em} \Phi(t+\Delta t, \tau) \, N \, \Phi(t+\Delta t, \tau) \,d\tau .
	\end{equation} %
	The notation $\Delta X(t)$ and $\Delta\mu(t)$ is convenient only when the sample time $\Delta t$ is fixed. Below we compute the minimum admissible code rate of a source code for $\Delta X(t)$ when $\Delta t$ can vary. Applying Theorem \ref{thm:code-rate} and Lemma \ref{lem:gaussian-source}, we see that the minimum code rate depends on how the covariance $\mathcal{W}_t(\Delta t)$ varies with $\Delta t$.
	
	\begin{prop} \label{prop:lti-code-rate-dx}
		Let $\Delta X(t) \sim N(\Delta\mu(t), \mathcal{W}_t(\Delta t))$, with covariance given by (\ref{eqn:noise-gramian}), be a source of discrete-time changes of state of system (\ref{eqn:stoch-lin-sys}) over the fixed interval $[t, t+\Delta t]$, with $t\geq 0$, $\Delta t > 0$. The minimum admissible code rate of a source code for $\Delta X(t)$ when $0 \leq D < n \min_i\{\lambda_i(\mathcal{W}_t(\Delta t))\}$ is
		\begin{equation} \label{eqn:complexity}
			R_{t,\Delta t}(D) := \frac{1}{2}\ln\det \mathcal{W}_t(\Delta t) - \frac{1}{2}\ln\det(\frac{D}{n}I)
		\end{equation}
		and otherwise obtained by reverse water-filling on the eigenvalues of $\mathcal{W}_t(\Delta t)$.
	\end{prop}

	In the time-invariant case, $A(t)=A$, the covariance $\mathcal{W}_t(\Delta t) =: \mathcal{W}(\Delta t)$ depends only on the sampling interval:
	\begin{equation} \label{eqn:noise-gramian-LTI}
		\mathcal{W}(\Delta t) = \int_0^{\Delta t} e^{A(\Delta t-\tau)} \, N \, e^{A^T(\Delta t - \tau)} \,d\tau
	\end{equation}
	As a result, the minimum admissible code rate $R_{t,\Delta t}(D) =: R_{\Delta t}(D)$ only varies with the sampling interval $\Delta t$. 
	The following result relates the stability of a LTI system to its minimum admissible code rate.
	
	\begin{prop} \label{prop:max-code-rate}
		If linear system (\ref{eqn:stoch-lin-sys}) is asymptotically stable and time-invariant, such that $A(t)=A$ is Hurwitz, then the rate distortion function for the source $\Delta X(t) \sim N(\Delta\mu(t), \mathcal{W}(\Delta t))$ is upper-bounded for all $D\geq 0$ as
		\begin{equation*}
			R_\infty(D) \geq R_{\Delta t}(D)
		\end{equation*}
		where $R_\infty(D)$ is the rate distortion function for the Gaussian source whose covariance $\mathcal{W}_\infty$ is the unique equilibrium solution to the continuous-time Lyapunov equation
		\begin{equation} \label{eqn:lyap-eqn}
			A\,\mathcal{W}_\infty + \mathcal{W}_\infty\,A^T + N = 0.
		\end{equation}
	\end{prop}
	\begin{proof}
		Applying the Leibniz rule to (\ref{eqn:noise-gramian-LTI}) yields the matrix differential equation
		\begin{equation} \label{eqn:W-diff-eqn}
			\frac{d \mathcal{W}(\Delta t) }{ d\Delta t } = A \mathcal{W}(\Delta t) + \mathcal{W}(\Delta t) A^T + N
		\end{equation}
		Since $A$ is Hurwitz, by \cite[Theorem~11.3]{brockettFiniteDimensionalLinear2015} there exists a unique, symmetric positive-definite equilibrium solution $\mathcal{W}_\infty$ to this differential equation. 
		For fixed $D$, \[ \tilde{R}_{\Delta t}(D) := \frac{1}{2}\ln\det \mathcal{W}(\Delta t) - \frac{1}{2}\ln\det(\frac{D}{n}I) \]
		is a smooth function of $\Delta t$ with derivative
		\begin{equation*}
			\dot{\tilde{R}}_{\Delta t}(D) := \frac{d R_{\Delta t}(D)}{d\Delta t} = \tr \Big( \mathcal{W}^{-1}(\Delta t) \frac{d \mathcal{W}(\Delta t)}{d\Delta t} \Big).
		\end{equation*}
		Thus, it has a unique equilibrium and is non-decreasing from 
		\begin{gather*}
			-\infty = \lim_{\Delta t \downarrow 0}\tilde{R}_{\Delta t}(D) \ \text{ to } \  R_\infty(D) := \lim_{\Delta t \to \infty} \tilde{R}_{\Delta t}(D).
		\end{gather*}
		It follows that $R_{\Delta t}(D) = \max(0, \tilde{R}_{\Delta t}(D))$ is also non-decreasing and asymptotically reaches $R_\infty(D)$. 
	\end{proof}

	\parsec{Attention and complexity} In \cite{brockettMinimumAttentionControl1997} a framework for the design of controllers capable of operating both with and without state feedback, called \textit{attention-varying control}, was developed under the premise that open-loop control functions admit simple algorithmic implementations, but closed-loop control functions require more complex implementations and computing resources, and are therefore undesirable when an open-loop controller will suffice. These ideas were further explored in \cite{baillieulFeedbackDesignsControlling1999, baillieulFeedbackCodingInformationbased2002}.

	Consider an application in which we are are to communicate the change in state of system (\ref{eqn:stoch-lin-sys}) over a memoryless channel with a fixed maximum information capacity $C$ in bits per channel use. The change in system state over a time interval $\Delta t$ is received by a control subsystem whose performance degrades unacceptably if the mean square error in $\Delta X(t)$ is above a prescribed limit $D$. Assume also that we are given a compressor for $\Delta X(t)$ that is efficient in the sense that for every $t$ and $\Delta t$, it operates with an expected average distortion $\bar{D}$ at a rate $\bar{R} = R_{t,\Delta t}(\bar{D}) + \epsilon$ for some small $\epsilon > 0$. By the \textit{source-channel separation theorem} \cite[Theorem~10.4.1]{coverElementsInformationTheory2005}, the mean square performance level $\bar{D}$ is achievable over the given channel with capacity $C$ if and only if $R_{t,\Delta t}(\bar{D}) + \epsilon < C$. %
	
	
	If system (\ref{eqn:stoch-lin-sys}) is time-invariant, the system matrix $A$ is also Hurwitz, and $R_\infty(\bar{D}) < C$ then by Proposition \ref{prop:max-code-rate} the control subsystem is free to operate at a sufficiently low rate: any sufficiently slow sampling rate $f_s:=1/\Delta t$ suffices to meet the distortion requirement, cf. case (a) in Figure \ref{fig:ex-lti-code-rates-dx}. The controller can operate in an essentially open-loop mode.

	On the other hand, suppose that the system matrix $A$ is not Hurwitz and does not have an equilibrium solution to the Lyapunov equation (\ref{eqn:lyap-eqn}). Then the minimum admissible code rate is unbounded with increasing $\Delta t$ and there exists a sampling rate below which the channel cannot support the controller's performance requirements. For example, fix the channel capacity at $8$ bits/use, and consider the control systems with code rates (complexities) shown in Figure \ref{fig:ex-lti-code-rates-dx}(b, c). By the source-channel separation theorem, in order to meet a maximum distortion requirement of $D=0.01$ 
	the sampling rate must be increased to at least $f_s = 1.63$ Hz for unstable system (c), and $f_s = 1.2$ Hz for system (b). This discussion motivates the following result.

	\begin{cor} \label{cor:sampling-rate-capacity}
		Let $C$ be the capacity of a given channel over which reproductions of the change in state $\Delta X(t) \sim N(\Delta\mu(t), \mathcal{W}(\Delta t))$ for a time-invariant system (\ref{eqn:stoch-lin-sys}), $A(t)=A$, are to be transmitted. If Lyapunov equation (\ref{eqn:lyap-eqn}) has no solution then there exists a sufficiently fast finite sampling rate $f_s < \infty$ for which $C > R_{1/f_s}(D)$, for any $D\geq 0$.
	\end{cor}
	\begin{proof}
		If there does not exist a solution $\mathcal{W}_\infty$ to (\ref{eqn:lyap-eqn}) then (\ref{eqn:W-diff-eqn}) has no equilibria and then $\tilde{R}_{\Delta t}(D)$ is strictly increasing without bound as $\Delta t\to\infty$. Then, either $R_{\Delta t}(D) = \max(0, \tilde{R}_{\Delta t}(D)) = 0$ for all $\Delta t\geq 0$, or there exists a finite nonzero $\Delta \tau$ for which $R_{\Delta t}(D) > 0$ for all $\Delta t \geq \Delta \tau$.
	\end{proof}%
%

	The source-channel separation theorem requires the control subsystem to increase its sampling rate in order to maintain a required mean square performance over the fixed capacity communications channel. If the rate of channel uses $f_s$ increases then, intuitively speaking, the control subsystem is ``more attentive'' to the system's change in state and operating in an essentially closed-loop mode. We propose $f_s$ as a discrete-time measure of ``attention'' consistent with the premise of attention-varying control \cite{brockettMinimumAttentionControl1997}.

	\begin{defn} \label{defn:complexity}
		The \textit{complexity} of the continuous-time linear system (\ref{eqn:stoch-lin-sys}) is the minimum admissible code rate given by the rate distortion function $R_{t, \Delta t}(D)$, equation (\ref{eqn:complexity}), of the source of uncertain changes in state $\Delta X(t) \sim N(\Delta\mu(t), \mathcal{W}_t(\Delta t))$.
	\end{defn}
	The complexity of a discrete-time system varies with 
	\begin{itemize}
		\item \textit{time} $t\geq 0$, unless of course (\ref{eqn:stoch-lin-sys}) is time-invariant,
		\item \textit{fidelity} or mean square distortion $D\geq 0$, and
		\item \textit{attention} or sampling rate $f_s=1/\Delta t$, $\Delta t> 0$.
	\end{itemize}

	\section{Emulating systems and source families}
	\label{sec:emu-sys}

	In this section we construct source codes for $\Delta X(t)$ using the endpoint map of a special class of control systems called \textit{emulating systems}. We define the code rate of an emulating system by way of example, and illustrate the capability of emulating systems for data-based, ``model-free'' simulation of unknown dynamical systems.
	
	\begin{defn}[Emulating systems] \label{defn:emu-sys}
		An \textit{emulating system} is the control system on $\R^n$ 
		\begin{equation} \label{eqn:emu-sys}
			\dot{x}(t) = \sum_{i=1}^K V_i(x(t)) u_i(t)
		\end{equation}
		associated with a \textit{source family} of $K$ autonomous vector fields $\mathcal{V}=\{V_i\}_{i=1}^K$ and an \textit{admissible control set} $\mathcal{U}$ of \textit{binary control functions} taking values on the set $\{0,1\}^K$ with finitely many switchings.
	\end{defn}
	
	


	Recall that the \textit{flow} of vector field $V_i$ is a diffeomorphism $\varphi^{(i)}_{\delta t}: \R^n \to \R^n$ mapping an initial condition $x(t)$ to the solution $x(t+\delta t)$ of the differential equation $\dot{x} = V_i(x)$.

	The \textit{endpoint map} of control system (\ref{eqn:emu-sys}) is the function $\phi: \R^n \times \R_{\geq 0} \times \mathcal{U} \to \R^n$ mapping an admissible control function $u \in \mathcal{U}$ to the solution $x(t+\Delta t)$ of system (\ref{eqn:emu-sys}) with initial condition $x(t)$:
	\begin{equation} \label{eqn:emu-sys-endpoint-map}
		\phi(x(t), \Delta t, u) = x(t) + \sum_{i=1}^K \int_t^{t+\Delta t}\hspace*{-1.5em} V_i(x(\tau))\, u_i(\tau)\, d\tau .
	\end{equation}
	
	The \textit{attainable set} from $x(t) \in \R^n$ is the subset of points reachable via the endpoint map from $x(t)$ by any admissible control function in $\mathcal{U}$ for some time $\tau \in [t, t+\Delta t]$. In the discussion below we find it convenient to refer instead to the set of ``attainable increments''
	\begin{equation*}
		\Att{x(t)}{\mathcal{U}, \mathcal{V}}(\Delta t) = \{ \phi(x(t), \tau, u)-x(t) \suchthat{} t \in [t, t + \Delta t], u \in \mathcal{U} \}.
	\end{equation*} %
	Lossless compression is possible if $\Delta x(t) \in \Att{x(t)}{\mathcal{U}, \mathcal{V}}(\Delta t)$.

	\subsection{The code rate of an emulating system}
	
	In \cite{sunNeuromimeticLinearSystems2022} the value of an admissible control function is called an \textit{activation pattern} and, as in Definition \ref{defn:emu-sys}, forms a discrete and finite set. An admissible control function changes its activation pattern at a discrete set of \textit{switching times}. Below we compute the code rate of two example emulating systems whose switching time sets admit parameterizations by discrete and finite data.

	\parsec{Elementary activation patterns} Let $\|\cdot\|_1$ denote the $L_1$ norm of a vector. Consider a source family $\mathcal{V}$ of $K$ autonomous vector fields with admissible controls \[ \mathcal{U} = \{ u(\cdot) \suchthat u(t) \in \{0,1\}^K, \, \|u(t)\|_1 \leq 1 \} \] consisting of functions whose activation patterns are elementary, or ``one-hot'', vectors. For each $u\in\mathcal{U}$ the endpoint map reduces to a composition of flows
	\begin{equation*}
		\phi(x(t), \Delta t, u) = \varphi^{(i_N)}_{\delta t_{i_N}} \circ \cdots \circ \varphi^{(i_1)}_{\delta t_{i_1}}(x(t))
	\end{equation*}
	where $0\leq N<\infty$ is the number of switchings of the control $u\in\mathcal{U}$, and $i_j\in[K]$ is the index of the vector field along which the system flows for time $\delta t_{i_j}$ before switching to the flow of $V_{i_{j+1}}$ for all $1\leq j\leq K-1$. For all $u\in\mathcal{U}$ restrict the switching times to occur uniformly on the boundaries of intervals $[k\tfrac{\Delta t}{N}, (k+1)\tfrac{\Delta t}{N}]$ for integers $0\leq k\leq N$, with the maximum number of switching times $N$ fixed and given. An algorithm that identifies points in the attainable set with ordered $N$-sequences of indices defines a compressor $f: \Att{x(t)}{\mathcal{V},\mathcal{U}}(\Delta t) \to [K]^N$, $f(\Delta x(t)) = (i_1, \ldots, i_N)=:I$. The decompressor $g:[K]^N\to\Att{x(t)}{\mathcal{V},\mathcal{U}}(\Delta t)$, 
	\begin{gather*}
		g(i_1, \ldots, i_N) = \varphi^{(i_N)}_{\Delta t/N} \circ \cdots \circ \varphi^{(i_1)}_{\Delta t/N}(x(t))
	\end{gather*}
	maps a sequence of indices to the attainable set. Let $\mathsf{U}:=[K]^N$ and $\mathsf{V}:=\Att{x(t)}{\mathcal{V},\mathcal{U}}(\Delta t)$. Defining suitable extensions $f^L:\mathsf{V}^L\to\mathsf{U}^L$, $g^L:\mathsf{U}^L\to\mathsf{V}^L$ to operate on blocks of size $L$ of iid observed state changes $\Delta x^L(t) = \{\Delta x^{(i)}(t)\}_{i=1}^L$ and $I^L:=\{I^{(i)}\}_{i=1}^L$ of $N$-sequences of indices would yield a $(K^N, L)$-source code for (\ref{eqn:stoch-lin-sys}) with code rate $R = \tfrac{N}{L}\log_2(K)$ bits/symbol. The integer $N$ is seen as an upsampling ratio from sample rate $f_s=1/\Delta t$ to $N/\Delta t$. Although the minimum admissible code rate for an LTI system decreases with $\Delta t$ (Corollary \ref{cor:sampling-rate-capacity}), the code rate of this emulating system does not depend on $\Delta t$ and is kept low by reducing $N$ and increasing the block size $L$. 


	%

	\parsec{Constant vector field approximations} An alternative source code with different de/compressor maps $f$, $g$ can be constructed from a source family $\mathcal{V}$ of constant vector fields, or of vector fields that are approximately constant for sufficiently small $\Delta t$. Let the admissible controls $\mathcal{U}$ be given as in Definition \ref{defn:emu-sys}. The endpoint map (\ref{eqn:emu-sys-endpoint-map}) for this emulating system reduces to
	\begin{equation} \label{eqn:emu-sys-endpoint-map-const}
		\Delta x(t) = \sum_{i=1}^K V_i(x(t)) \int_t^{t+\Delta t}\hspace*{-1.5em} u_i(\tau)\, d\tau 
	\end{equation}
	Let $\delta t_i := \int_t^{t+\Delta t} \hspace*{-1.5em} u_i(\tau)\, d\tau$ denote the total time spent flowing along vector field $V_i\in\mathcal{V}$, and $p_i(t) := \delta t_i/Z(t)$ the relative fraction of time, where $Z(t):=\sum_{i=1}^K \delta t_i \leq \Delta t$. We have thus defined a decompressor $g: \Delta^K \times \R_{\geq0} \to \Att{x(t)}{\mathcal{V},\mathcal{U}}(\Delta t)$ mapping a ``normalizing constant'' $Z(t)$ and a probability vector $p(t) \in \Delta^K := \{ p\in\R^K \suchthat \sum_{i=1}^K p_i = 1, \, p_i \geq 0\}$ to the attainable set,
	\begin{equation} \label{eqn:emu-sys-endpoint-map-simplex}
		\Delta \hat{x}(t) = g(p(t), Z(t)) = Z(t) \sum_{i=1}^K V_i(x(t)) \, p_i(t) .
	\end{equation} %
	In Subsection \ref{sec:emu-sys-data-based} we define a compressor $f: \Att{x(t)}{\mathcal{V},\mathcal{U}}(\Delta t) \to \Delta^K \times \R_{\geq0}$ as the solution of a linear program, and use the maps $f,g$ to simulate unknown dynamical systems from observations.

	Now fix an integer $N > 0$ and consider $K$ nonnegative integers $n(t) := (n_1(t), \ldots, n_K(t))$ such that $N = \sum_{i=1}^K n_i(t)$. Viewing $n_i(t)/N$ as a rational approximation of $p_i(t)$ and assuming $Z(t)=\Delta t$, we obtain a reproduction $\Delta\hat{x}(t) = g\circ f(\Delta x(t))$ where the compressor assigns $K$ nonnegative integers to each observed change of state, $f(\Delta x(t)) = n(t)$, and the decompressor maps the $K$ integers to the attainable set $g(n(t)) = \Delta t \sum_{i=1}^K V_i(x(t)) \, n_i(t)/N$. The code rate of this source code depends on the number of ways of uniquely matching an integer $n_i(t)$ with a vector field $V_i \in \mathcal{V}$. 
	
	\begin{figure}[t!]
		\centering
		\includegraphics[width=.775\linewidth]{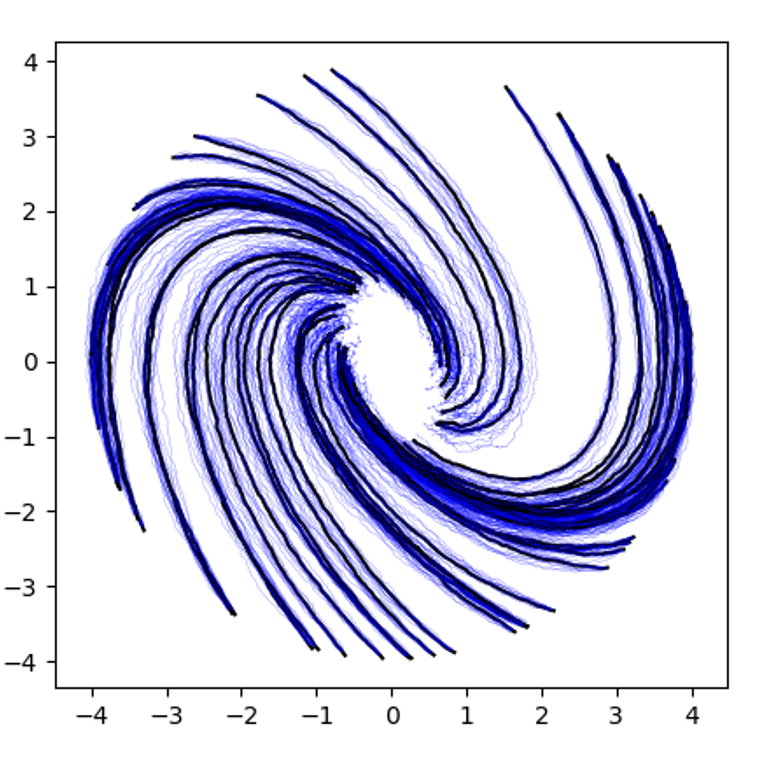} \vspace*{-1em}
		\caption{ Sample paths (blue) of a time-invariant stochastic linear system with system matrix (a) in Figure \ref{fig:ex-lti-code-rates-dx} are plotted against trajectories generated by the multinomial sampling scheme (black) enabled by identification of points in the attainable set with points in the standard $K$-simplex. At this scale the emulated sample paths are visually indistinguishable from the training data.}
		\label{fig:emu-sys-gend-trajs}
		\vspace*{-1em}
	\end{figure}

	\subsection{Data-based emulation of unknown systems} \label{sec:emu-sys-data-based}
	
	Suppose we conduct an experiment consisting of $L$ independent trials culminating in a 
	``training dataset'' \[ \mathcal{X} = \{ x^{(i)}(k\Delta t) \given{} k\in[T], \, i\in[L] \} \] of observations of an unknown dynamical system. Our task is to 
	generate a new trajectory $\{\tilde{x}(k\Delta t) \}_{k=0}^T$ that, in some sense, would closely resemble the result of an $(L+1)$\textsuperscript{th} trial.
	
	We propose an approach based on multinomial sampling, assuming $\Delta x^{(i)}(t) \in \Att{x^{(i)}(t)}{\mathcal{V},\mathcal{U}}(\Delta t)$ for all $i\in[L]$. First define a compressor $f(\Delta x(t)) = (\delta t^*/Z(t), Z(t))$ using the solution of the linear program \vspace*{-.5em}
	\begin{equation*} 
		\delta t^* = \arg\min_{\delta t\geq0} \: \|\delta t \|_1 \ \ 
		\text{s.t.} \ \  \Delta x(t) = \sum_{j=1}^K V_j\,\delta t_j \vspace*{-.25em} 
	\end{equation*} %
	where $Z(t)=\sum_j \delta t^*_j$. 
	At each $t \in \{k\Delta t\}_{k=1}^T$, generate a new $\Delta \tilde{x}(t)$ as follows:
	\begin{enumerate}
		\item Compute $(p^{(i)}(t), Z^{(i)}(t)) = f(\Delta x^{(i)}(t))$ for each $i\in[L]$, and average: $\tilde{p}(t)=\tfrac{1}{L}\sum_{i=1}^L p^{(i)}(t)$, $\tilde{Z}(t)=\tfrac{1}{L}\sum_{i=1}^L Z^{(i)}(t)$.
		\item Sample $\tilde{n}(t) \sim \mathsf{Mult}(n \given{} N, \tilde{p}(t))$ from the multinomial distribution for $N$ independent trials with $K$ outcomes; $\tilde{p}_j(t)= \mathsf{Pr}(V(t) = V_j)$ is the probability of selecting $V_j\in\mathcal{V}$ in one trial.
		\item Compute $\Delta\tilde{x}(t):=g(\tilde{n}_i(t)/N, \tilde{Z}(t))$ using (\ref{eqn:emu-sys-endpoint-map-simplex}).
	\end{enumerate}
	We call this scheme ``non-parametric'' because the generation of new trajectories does not require us to tune or identify system parameters such as poles, zeros, or process noise covariances, nor integrate any differential equations. Figure \ref{fig:emu-sys-gend-trajs} shows the performance of this scheme on training data collected from a stable LTI system. The emulating system (5) of \cite{sunNeuromimeticLinearSystems2022} has as its source family a set of 24 constant ``vector fields'' \[ \mathcal{V} = \Big \{ \vctr{-2 \\ -2}, \vctr{-2\\-1}, \cdots, \vctr{2\\1}, \vctr{2\\2} \Big\}. \] The training data $\mathcal{X}$ consists of sample paths of the time-invariant stochastic linear system with process noise intensity $N = 0.01I_2$ and the stable $A$ matrix of Figure \ref{fig:ex-lti-code-rates-dx}(a). Trajectories are emulated for 3 seconds at a sampling rate of $f_s = 100$ Hz. With admissible mean square distortion $D=0.01$ 
	the minimum required code rate (Proposition \ref{prop:lti-code-rate-dx}) is 2 bits/symbol per channel use, or a data rate of 200 bits/symbol/sec. %

	\section{Conclusion}
	We have defined the complexity of a sampled data representation of a linear stochastic system to be the minimum admissible code rate of a source code for its forward increments. The complexity is a quantity of relevance in applications requiring local and dynamic decision-making. In a context requiring communication of the change in state of a given linear stochastic control system over a channel of fixed capacity we proposed the minimum sampling rate required to lower the system complexity below the channel capacity as a measure of the degree to which an ``attention-varying controller'' \cite{brockettMinimumAttentionControl1997} could operate with or without feedback. We constructed explicit source codes from the endpoint maps of emulating systems, and illustrated their use in data-based, non-parametric simulation and analysis of unknown dynamical systems. Further applications of emulating systems to estimation and control tasks in both the infinitesimal setting \cite{weissmanDirectedInformationCausal2013} and in a sequential encoding context \cite{stavrouTimeInvariantMultidimensionalGaussian2020} will appear in future work.

	\bibliographystyle{plain}
	\bibliography{refs}

\begin{thebibliography}{10}

\bibitem{baillieulFeedbackCodingInformationbased2002}
J.~Baillieul.
\newblock Feedback coding for information-based control: Operating near the
  data-rate limit.
\newblock In {\em Proceedings of the 41st {{IEEE Conference}} on {{Decision}}
  and {{Control}}, 2002.}, volume~3, pages 3229--3236, {Las Vegas, NV, USA},
  2002. {IEEE}.

\bibitem{baillieulFeedbackDesignsControlling1999}
John Baillieul.
\newblock Feedback {{Designs}} for {{Controlling Device Arrays}} with
  {{Communication Channel Bandwidth Constraints}}.
\newblock In {\em Fourth {{ARO Workshop}} on {{Smart Structures}}}, page~7.
  {Penn State}, {University Park, PA}, August 1999.

\bibitem{bergerRateDistortionTheory2003}
Toby Berger.
\newblock {\em Rate-{{Distortion Theory}}}.
\newblock Information and {{System Sciences}}. {John Wiley \& Sons, Inc.},
  {Hoboken, NJ, USA}, April 2003.

\bibitem{brockettMinimumAttentionControl1997}
Roger~W. Brockett.
\newblock Minimum attention control.
\newblock In {\em Proceedings of the 36th {{IEEE Conference}} on {{Decision}}
  and {{Control}}}, volume~3, pages 2628--2632, {San Diego, CA, USA}, 1997.
  {IEEE}.

\bibitem{brockettFiniteDimensionalLinear2015}
Roger~W. Brockett.
\newblock {\em Finite {{Dimensional Linear Systems}}}.
\newblock {Society for Industrial and Applied Mathematics}, {Philadelphia, PA},
  May 2015.

\bibitem{brockettQuantizedFeedbackStabilization2000}
R.W. Brockett and D.~Liberzon.
\newblock Quantized feedback stabilization of linear systems.
\newblock {\em IEEE Transactions on Automatic Control}, 45(7):1279--1289, July
  2000.

\bibitem{coverElementsInformationTheory2005}
Thomas~M Cover and Joy~A Thomas.
\newblock {\em Elements of {{Information Theory}}}.
\newblock {John Wiley \& Sons, Ltd}, 2005.

\bibitem{jazwinskiStochasticProcessesFiltering1970}
Andrew~H. Jazwinski.
\newblock {\em Stochastic {{Processes}} and {{Filtering Theory}}}.
\newblock {Academic Press, Inc.}, {New York}, 1970.

\bibitem{nairFeedbackControlData2007}
Girish~N. Nair, Fabio Fagnani, Sandro Zampieri, and Robin~J. Evans.
\newblock Feedback {{Control Under Data Rate Constraints}}: {{An Overview}}.
\newblock {\em Proceedings of the IEEE}, 95(1):108--137, January 2007.

\bibitem{nairTopologicalFeedbackEntropy2004}
G.N. Nair, R.J. Evans, I.M.Y. Mareels, and W.~Moran.
\newblock Topological {{Feedback Entropy}} and {{Nonlinear Stabilization}}.
\newblock {\em IEEE Transactions on Automatic Control}, 49(9):1585--1597,
  September 2004.

\bibitem{stavrouTimeInvariantMultidimensionalGaussian2020}
Photios~A. Stavrou, Takashi Tanaka, and Sekhar Tatikonda.
\newblock The {{Time-Invariant Multidimensional Gaussian Sequential
  Rate-Distortion Problem Revisited}}.
\newblock {\em IEEE Transactions on Automatic Control}, 65(5):2245--2249, May
  2020.

\bibitem{sunNeuromimeticLinearSystems2022}
Zexin Sun and John Baillieul.
\newblock Neuromimetic {{Linear Systems}} -- {{Resilience}} and {{Learning}}.
\newblock {\em arXiv e-prints}, 2022.

\bibitem{tatikondaControlCommunicationConstraints2004}
S.~Tatikonda and S.~Mitter.
\newblock Control {{Under Communication Constraints}}.
\newblock {\em IEEE Transactions on Automatic Control}, 49(7):1056--1068, July
  2004.

\bibitem{tatikondaStochasticLinearControl2004}
S.~Tatikonda, A.~Sahai, and S.~Mitter.
\newblock Stochastic {{Linear Control Over}} a {{Communication Channel}}.
\newblock {\em IEEE Transactions on Automatic Control}, 49(9):1549--1561,
  September 2004.

\bibitem{weissmanDirectedInformationCausal2013}
T.~Weissman, {Young-Han Kim}, and H.~H. Permuter.
\newblock Directed {{Information}}, {{Causal Estimation}}, and
  {{Communication}} in {{Continuous Time}}.
\newblock {\em IEEE Transactions on Information Theory}, 59(3):1271--1287,
  March 2013.

\bibitem{wingshingwongSystemsFiniteCommunication1999}
{Wing Shing Wong} and R.W. Brockett.
\newblock Systems with finite communication bandwidth constraints. {{II}}.
  {{Stabilization}} with limited information feedback.
\newblock {\em IEEE Transactions on Automatic Control}, 44(5):1049--1053, May
  1999.

\end{thebibliography}
\end{document}